\documentclass[12pt,twoside]{mitthesis}
\usepackage{lgrind}
\usepackage{amsmath}
\usepackage{amssymb}
\usepackage{amsthm}
\usepackage{cmap}
\usepackage[T1]{fontenc}
\usepackage{setspace}
\usepackage{graphicx}

\newcommand{\Mod}[1]{\ \text{mod}\ #1}
\newcommand{\cswap}{\emph{CSWAP }}
\newcommand{\ccswap}{\emph{CCSWAP }}
\newcommand{\cccswap}{\emph{CCCSWAP }}
\newcommand{\ckswap}{\emph{C\textsuperscript{k}SWAP }}
\newcommand{\ckkswap}{\emph{C\textsuperscript{k-1}SWAP }}
\newcommand{\cmmswap}{\emph{C\textsuperscript{m-3}SWAP }}
\newcommand{\ciswap}{\emph{C\textsuperscript{i}SWAP }}
\newcommand{\ciiswap}{\emph{C\textsuperscript{i-1}SWAP }}
\newcommand{\cknot}{\emph{C\textsuperscript{k}NOT }}
\newcommand{\cfnot}{\emph{C\textsuperscript{4}NOT }}
\newcommand{\cinot}{\emph{C\textsuperscript{i}NOT }}
\newcommand{\ciinot}{\emph{C\textsuperscript{i-1}NOT }}
\newcommand{\cnnot}{\emph{C\textsuperscript{n}NOT }}
\newcommand{\cnnnot}{\emph{C\textsuperscript{n-1}NOT }}
\newcommand{\T}{$t'_1$ }
\newcommand{\TT}{$t'_2$ }

\begin{document}

\newtheorem{Def}{Definition}
\newtheorem{Lem}{Lemma}
\newtheorem{Cla}{Claim}

%
%
%
%
%
%
%
%
%
%
%

\title{Reversible Logic Synthesis with Minimal Usage of Ancilla Bits}

\author{Siyao Xu}
\department{Department of Electrical Engineering and Computer Science}

\degree{Master of Engineering in Electrical Engineering and Computer Science}

\degreemonth{June}
\degreeyear{2015}
\thesisdate{May 20, 2015}


\supervisor{Prof. Scott Aaronson}{Associate Professor}

\chairman{Prof. Albert R. Meyer}{Chairman, Masters of Engineering Thesis Committee}

%
%
%
%
%
%
%
%
%
%
%
%
%
%
%
%
%
%
%
%
%
\maketitle



\cleardoublepage
\setcounter{savepage}{\thepage}
\begin{abstractpage}
%
%
%

Reversible logic has attracted much research interest over the last few decades, especially due to its application in quantum computing. In the construction of reversible gates from basic gates, ancilla bits are commonly used to remove restrictions on the type of gates that a certain set of basic gates generates. With unlimited ancilla bits, many gates (such as Toffoli and Fredkin) become universal reversible gates. However, ancilla bits can be expensive to implement, thus making the problem of minimizing necessary ancilla bits a practical topic.

This thesis explores the problem of reversible logic synthesis using a single base gate and a few ancilla bits. Two base gates are discussed: a variation of the 3-bit Toffoli gate and the original 3-bit Fredkin gate. There are three main results associated with these two gates: i) the variated Toffoli gate can generate all $n$-bit reversible gates using 1 ancilla bit, ii) the variated Toffoli can generate all $n$-bit reversible gates that are even permutations using no ancilla bit, iii) the Fredkin gate can generate all $n$-bit conservative reversible gates using 1 ancilla bit. Prior to this paper, the best known result for general universality requires three basic gates, and the best known result for conservative universality needs 5 ancilla bits.

The second result is trivially optimal. For the first and the third result, we explicitly prove their optimality: the variated Toffoli cannot generate all $n$-bit reversible gates without using any extra input lines, and the Fredkin gate cannot generate all $n$-bit conservative reversible gates without using extra input lines. We also explore a stronger version of the second converse by introducing a new concept called borrowed bits, and prove that the Fredkin gate cannot generate all $n$-bit conservative reversible gates without ancilla bits, even with an unlimited number of borrowed bits.

\end{abstractpage}


\singlespacing

\section*{Acknowledgments}

First of all, I would like to express my sincere gratitude towards Prof. Scott Aaronson, my M.Eng. supervisor, who guided me in my research and provided me with a lot of useful resources and advice. I would also like to thank Prof. Dana Moshkovitz, my academic advisor, for helping me find this project that I love.

I would like to thank Qian Yu, my fianc\'{e}, for teaching me the basics of quantum computing, and inspiring me on important techniques used in this thesis. Most importantly, Qian provided me with mental support without reservation, which helped me through difficult times during my research.

In addition, I thank all my friends for making my last year at MIT fun and meaningful.

Last but not least, my greatest gratitude goes to my parents for always respecting my decisions and providing me with unwavering support. 

\doublespacing

\cleardoublepage


\pagestyle{plain}

\tableofcontents
\newpage
\listoffigures

\chapter{Introduction}

\section{Reversible Gates}

\subsection{Definition and Representation}
In logic circuits, a reversible gate is defined as a gate where no information is lost in the process of computing the output from the input. Strictly speaking, an $n$-bit reversible gate takes in $n$ bits of input and returns $n$ bits of output, and when given the output, we can unambiguously reconstruct the input from that information. For example, the \emph{NOT} gate, which implements logical negation, is reversible because the input can be reconstructed by inverting the output. The \emph{AND} gate, which implements logical conjunction, is not reversible because 4 different inputs produce only 2 different outputs; specifically the output 0 has 3 different inputs mapped to it.

There are many ways to represent reversible gates. In this paper, we use mainly the following two forms of representation: \emph{permutation} and \emph{sum of products}.

\textbf{Permutation:} We consider all the possible inputs to an $n$-bit reversible gate as a set of binary strings. This set $I_n = \{0, 1\}^n$ contains $2^n$ elements. Since the output of an $n$-bit reversible gate also consist of $n$ bits, the set of all possible outputs is also $I_n$.

In order to allow for unambiguously reconstructing the input, the inputs must be mapped to the outputs one-on-one. In other words, an $n$-bit reversible gate can be represented as a permutation of the set $I_n$. As a result, there are $(2^n)!$ $n$-bit reversible gates.

Permutations have all kinds of properties that may be useful for classifying different reversible gates. In particular, we will be using the parity and the cycle representation of permutations in this paper to prove several results about reversible gate constructions.

\textbf{Sum of Products:} A gate (not necessarily reversible) can be represented as an arithmetic expression on the field of $\mathbb{F}_2$. An \emph{AND} gate can be represented as $a\cdot b$ where $a$ and $b$ are the two input bits. A \emph{XOR} gate can be represented as $a + b$. A \emph{NOT} gate is $a + 1$, where $a$ is the input bit. Because any gate can be constructed by cascading a certain number of \emph{AND} and \emph{XOR} gates in some order, we can translate the cascading sequence into an arithmetic expression. After expanding it, we will get a representation of the gate in the form of a sum of products.

Since all the operations happen on $\mathbb{F}_2$, certain properties hold true apart from associativity and commutativity laws for regular addition and multiplication.

\begin{Lem}
In the gate representation using the sum of products, the following hold true:
\begin{enumerate}
\item[(1)] $a+a=0$
\item[(2)] if $a \neq b$, $ab = 0$
\item[(3)] $a \cdot a = a$
\end{enumerate}
\end{Lem}

The representation using the sum of products is useful when we want a compact way of expressing the reversible gate output. It usually takes a more concise form than a complete truth table.

Apart from the two forms we just introduced, other common ones exist such as circuit representations, but for complicated gates they are much more difficult to analyze.

\subsection{Background}

Reversible logic provides a large potential in energy-efficient computers. Landauer's Principle \cite{landauer61} provides a lower bound on the energy loss for each bit of information that is erased during computation. Since no information is erased in reversible computing, devices built using purely reversible circuits can potentially achieve energy efficiency greater than that bounded using Landauer's Principle. Such observation motivated research in reversible computing. Furthermore, recent decades have seen the flourishing of quantum computing, which attracted more research interest in reversible computing since reversible circuits serve as important components in many quantum algorithms.

Among the most extensively studied problems in reversible computing, one of the most fundamental is the question of \emph{universality}. In other words, we are interested in whether all reversible gates can be constructed using certain basic reversible gates. We formally define universality as follows.  

\begin{Def}
Denote the universe of all reversible gates by $\mathcal{U}$. Let $S$ be a small subset of $\mathcal{U}$. For each gate $g$ in $\mathcal{U}$, we aim to represent it as a cascade of a sequence of gates $(g_1, g_2, ... , g_k)$, where each $g_i$ is a gate in $S$. If we are able to find such representations for all gates in $\mathcal{U}$, then we call the subset $S$ \emph{universal}.
\end{Def}

Just like \emph{NAND} and \emph{NOR} in the set of general gates, we are also interested in finding a small single reversible gate that is universal in $\mathcal{U}$. However, we can prove that no single reversible gate with a constant number of input bits can be universal without additional conditions. (See Appendix A) One way to make a non-universal gate universal, however, is to provide extra input lines called \emph{ancilla bits}. 

\section{Ancilla Bits}

An ancilla bit provides extra computational power to a gate by serving as an extra input line with a fixed input. For example, a 3-bit gate could request an ancilla bit with input 1, allowing itself to operate on 4 input bits in total. The ancilla bit can have its value changed in intermediate steps, but to maintain the reversibility of the gate, the ancilla bit has to be returned to its original value (1 in this example) as part of the output. An ancilla bit that is not guaranteed to be reset to its original value is said to produce a \emph{garbage output bit}.

\begin{Def}
Let $g$ be a $k$-bit reversible gate in $\mathcal{U}$, where $k$ is a constant. Let $\mathcal{U}_n$ denote the set of all $n$-bit reversible gates. For each gate $g_n$ in $\mathcal{U}_n$, we aim to find a cascade of $g$ operating on different input lines. In addition, we are allowed to use $b$ ancilla bits with constant values $(c_1, c_2, ... , c_b)$ where $c_i \in \{0, 1\}$. The cascade of $g$ should be equivalent to $g_n$ disregarding the ancilla output lines. In other words, if $g_n$ maps input $(i_1, i_2, ... , i_n)$ to output $(o_1, o_2, ... , o_n)$, then the cascade of $g$ plus $b$ ancilla bits should map input $(i_1, i_2, ... , i_n, c_1, c_2, ... , c_b)$ to output $(o_1, o_2, ... , o_n, c_1, c_2, ... , c_b)$. If such a cascade can be found for any gate in $\mathcal{U}_n$, we call the gate $g$ \emph{$n$-bit universal with $b$ ancilla bits}. 
\end{Def}

Many 3-bit reversible gates are $n$-bit universal given sufficient ancilla bits, allowing for garbage output bits. Two common examples are the \emph{Toffoli gate} and the \emph{Fredkin gate}.

\section{Toffoli Gate}

The Toffoli gate, or the \emph{Controlled-Controlled-Not} (\emph{$C^2NOT$}) gate, was first invented by Toffoli \cite{toffoli80}. It is $n$-bit universal given unlimited ancilla bits. Given three input bits $(a, b, c)$, the Toffoli gate inverts $c$ if and only if both $a$ and $b$ are 1, and it keeps $a$ and $b$ unchanged. Figure \ref{tof} shows the circuit and the truth table representations of the Toffoli gate. 

\begin{figure}[ht!]
\centering
\includegraphics[width=90mm]{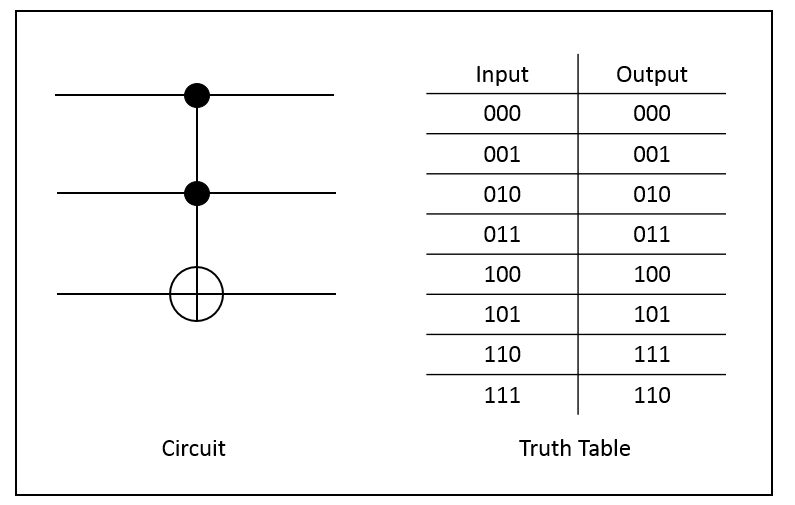}
\caption{Circuit and truth table representations of the Toffoli gate. \label{tof}}
\medskip
\small
In the circuit representation, the top two lines are called the \emph{control lines}, and the bottom line is called the \emph{target line}.
\end{figure}

We can also represent Toffoli gate using the sum of product representation mentioned in Section 1.1.1. $ab$ is equal to 1 if and only if both $a$ and $b$ are 1, so the Toffoli gate will return output $(a, b, ab+c)$.

By adding more control lines to the Toffoli gate, we get \emph{multi-control} Toffoli gates. A multi-control Toffoli gate with $k$ control lines (\cknot) takes in input $(a_1, a_2, ... , a_k, b)$, keeps the first $k$ bits unchanged, and inverts the last bit $b$ if and only if all the first $k$ bits are 1. Using the same sum of products representation, \cknot produces output $(a_1, a_2, ... , a_k, b + \prod_{i=1}^{k}{a_i})$.

\section{Fredkin Gate and Conservative Gates}

Another common gate that is $n$-bit universal with sufficient ancilla bits is the Fredkin gate, also know as \cswap, invented by Fredkin \cite{fredkin82}. Given three input bits $(a, b, c)$, the Fredkin gate swaps $b$ and $c$ if and only if $a$ is 1. It keeps $a$ unchanged. Figure \ref{fred} shows the circuit and the truth table representations of the Fredkin gate.

\begin{figure}[ht!]
\centering
\includegraphics[width=90mm]{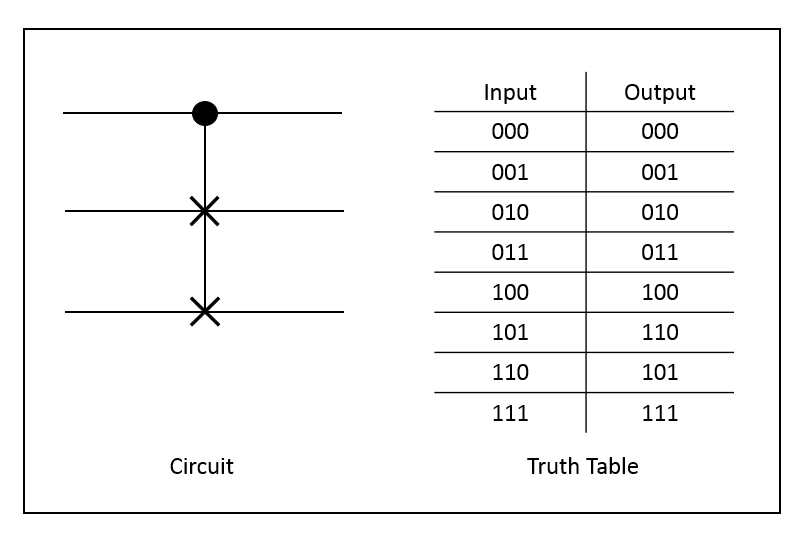}
\caption{Circuit and truth table representations of the Fredkin gate. \label{fred}}
\medskip
\small
In the circuit representation, the top line is called the \emph{control line}, and the bottom two lines are called the \emph{target lines}.
\end{figure}

We can also represent the Fredkin gate using the sum of products representation. Input $(a, b, c)$ corresponds to output $(a, (1-a)b + ac, (1-a)c + ab)$, which simplifies to $(a, b+ab+ac, c+ab+ac)$.

By adding more control lines to the Fredkin gate, we get \emph{multi-control} Fredkin gates. A multi-control Fredkin gate with $k$ control lines (\ckswap) takes in input $(a_1, a_2, ... , a_k, b, c)$, keeps the first $k$ bits unchanged, and swaps the last 2 bits $b$ and $c$ if and only if all the first $k$ bits are 1. Using the same sum of products representation, \ckswap produces output $(a_1, a_2, ... , a_k, b (1-\prod_{i=1}^{k}{a_i}) + c \prod_{i=1}^{k}{a_i}, c (1-\prod_{i=1}^{k}{a_i}) + b \prod_{i=1}^{k}{a_i})$.

An interesting property of the Fredkin gate is that it performs a conditional swap on the input bits, thus not changing the Hamming weight of the input string. We call this type of gate \emph{conservative}. Conservative gates have the following property which is trivial to prove:

\begin{Lem}
If a gate $g$ can be constructed as the cascade of a sequence of conservative gates, then gate $g$ must also be conservative. 
\end{Lem}

This lemma indicates that any gate constructed by cascading Fredkin gates and multi-control Fredkin gates must also be conservative. When using ancilla bits we need to return them to their original values, so even with unlimited ancilla bits, we are still unable to construct non-conservative gates using Fredkin gates without producing any garbage output.

Another question to ask about Fredkin gates is whether they are universal within the set of conservative gates, i.e. whether any $n$-bit conservative gate can be constructed using Fredkin gates with sufficient ancilla bits. This problem has practical applications in conservative logic synthesis and other topics. The answer is that Fredkin gate is $n$-bit universal in the subset of conservative gates using sufficient ancilla bits, without producing garbage output. \cite{fredkin82}

\section{Problem Statement and Related Work}

Ancilla bits help remove certain constraints on the type of gates that a base gate can generate. For example, cascading Toffoli gates cannot generate any $n$-bit gate that maps the all-zero input to a non-all-zero output, but it can do so by using ancilla bits with values set to 1.

However, ancilla bits have high implementation cost, thus making minimizing the necessary number of ancilla bits a practical problem. This thesis explores two main questions in this category: What is the minimum number of ancilla bits needed in order for a small reversible gate to be $n$-bit universal? What is the minimum number of ancilla bits needed for a small conservative reversible gate to be $n$-bit universal in the set of all conservative reversible gates?

The topic of reversible logic synthesis has been attracting research interest for decades. Due to the difference in their applications, various synthesis methods were proposed using different optimization metrics and base gate constraints. Toffoli \cite{toffoli80} first proved that the Toffoli gate is universal, using a recursive construction on the multi-control Toffoli gates. However, this method is not optimized. The paper also proved that, when using only the Toffoli gate as the base gate, a minimum number of 2 ancilla bits are needed just to generate the \emph{NOT} gate. Shende et al. \cite{shende02} later provided a construction for any even permutation reversible gate using no ancilla bit, with the NCT library (\emph{NOT}, \emph{CNOT}, Toffoli) as the set of base gates. However, this result uses three gates for construction instead of one. Aaronson et al. \cite{aaronson15} proposed a new construction for any reversible gate using only the Toffoli gate with 3 ancilla bits, and a construction for any conservative reversible gate using the Fredkin gate and 5 ancilla bits.

Apart from minimizing the number of ancilla bits, two other commonly used metrics are gate complexity and circuit depth \cite{saeedi10}. Various synthesis methods have been proposed in order to minimize gate complexity \cite{donald08,gupta06,maslov07,miller03}. These methods usually use multiple base gates as the generating set, for example, the NCT library, and sometimes the NCTSFP library (NCT plus \emph{SWAP}, Fredkin, and Peres \cite{peres85}). Also, the number of ancilla bits is not minimized in these synthesis methods.

\section{Organization}

This thesis focuses on minimizing the number of ancilla bits needed by any single base gate. Chapter 2 introduces a construction of any reversible gate using a variation of the Toffoli gate with 1 ancilla bit. Compared to the construction by Shende et al., this construction needs only a single gate as the base gate. We also define a new concept called the borrowed bit, which is a weaker version of the ancilla bit, and argue that 1 borrowed bit is enough for the above construction. Chapter 3 introduces a construction of any even permutation reversible gate using the variation of the Toffoli gate with no extra input line at all. Chapter 4 introduces a construction of any conservative reversible gate using the Fredkin gate with 1 ancilla bit. This improves the currently best known result \cite{aaronson15} by 4 ancilla bits. We also prove that, in terms of ancilla bits needed, 1 ancilla bit is the best we can achieve, and there exists no construction using only borrowed bits. Chapter 5 concludes the results and poses questions that remain to be solved on this topic.

\chapter{General Universality With One Borrowed Bit}

\section{Toffoli Gate and Universality}

There are several reasons why the Toffoli gate cannot be universal without ancilla bits. First, just like any small gate with a constant number of input bits, Toffoli is an even permutation when considered as an $n$-bit gate that only operates on 3 bits and leaves the rest unchanged. (See Appendix A) This means that it cannot generate odd permutations. In addition to that, Toffoli maps input $(0, 0, 0)$ to output $(0, 0, 0)$, so without any ancilla bits, any cascade using only Toffoli gates must also map the all-zero input to the all-zero output.

The first restriction applies to any $k$-bit gate where $k$ is smaller than $n$. The second restriction, however, is specific to the Toffoli gate. We can easily modify the Toffoli gate to remove this restriction.

\section{The Variated Toffoli Gate}

We invent a variation of the Toffoli gate by cascading the Toffoli gate with a \emph{NOT} on the second bit afterwards, and name it the \emph{variated Toffoli gate}. Figure \ref{var} shows its circuit and truth table representations. We can also represent it using the sum of products: It maps input $(a, b, c)$ to output $(a, b+1, ab+c)$.

\begin{figure}[ht!]
\centering
\includegraphics[width=90mm]{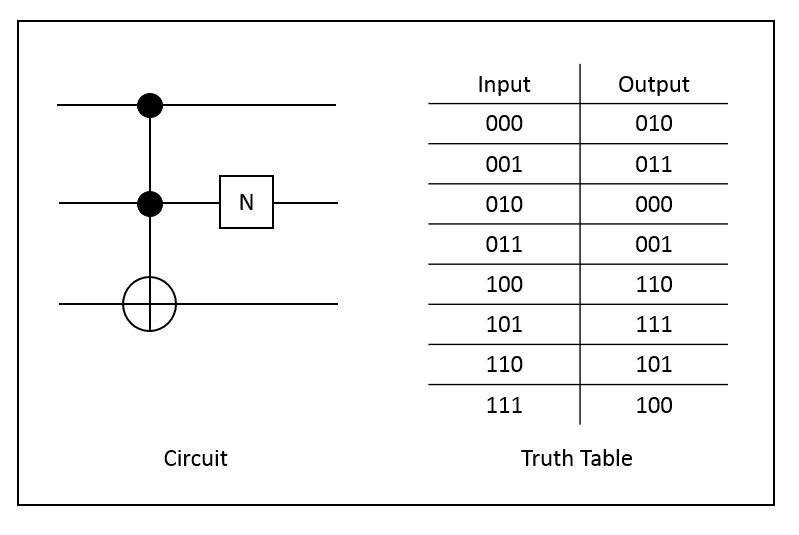}
\caption{The circuit and truth table representations of the variated Toffoli gate. \label{var}}
\end{figure}

From the truth table, it is not hard to see that the variated Toffoli gate is also a reversible gate. In addition, it maps the $(0, 0, 0)$ input to $(0, 1, 0)$. We claim that the variated Toffoli gate is $n$-bit universal with 1 ancilla bit. The construction comes in two steps: First we prove that any $n$-bit reversible gate can be constructed using a set of multi-control Toffoli gates. Then we provide a construction for any multi-control Toffoli gate using the variated Toffoli gate.

\section{Reversible Gate from \cinot}

We start by constructing any reversible gate using certain multi-control Toffoli gates.

\begin{Cla}
Let $T_n$ denote the set of gates ${NOT, CNOT, CCNOT, ... , \emph{\cnnot}}$. Given any $n$-bit reversible transformation $g$, there exists a construction of $g$ by cascading gates in $T_n$.
\end{Cla}

Recall that in the permutation representation of a gate $g$, we find an equivalent permutation $\pi$ on the input set $I_n$, which consists of all binary inputs of length $n$. If we consider these binary strings as numbers, then equivalently $\pi$ is a permutation on the set $N$ of numbers $\{0, 1, 2, ... , 2^n-1\}$. The following is a well-known result about decomposing this permutation into a sequence of basic transformations.

\begin{Lem}
Any permutation $\pi$ on the set $N = \{0, 1, 2, ... , 2^n-1\}$ can be decomposed into the following two basic transformations:
\begin{enumerate}
\item[(1)] $t_1$: swapping 0 and 1 
\item[(2)] $t_2$: shifting $k$ to $(k+1) \Mod 2^n$, $\forall k \in N$
\end{enumerate}
\end{Lem} 

This lemma reduces the problem of constructing any reversible gate to constructing the $t_1$ and $t_2$ gates using multi-control Toffoli gates. A proof of this lemma is included in Appendix B.

\begin{Lem}
The $n$-bit gates $t_1$ and $t_2$ can be constructed by cascading gates from $T_n$.
\end{Lem}

\begin{proof}
The gate $t_1$ swaps elements $0$ and $1$, which, in the original binary settings, is equivalent to mapping input $(0, 0, ... , 0)$ to $(0, 0, ... , 1)$, mapping input $(0, 0, ... , 1)$ to $(0, 0, ... , 0)$, and mapping all other inputs to themselves. This is equivalent to inverting the last bit if and only if all first $n-1$ bits are $0$, which is exactly the opposite from the multi-control Toffoli gate. Thus, $t_1$ can be constructed by applying \emph{NOT} to each of the first $n-1$ bits, followed by a \cnnnot with the last bit being the target line, and finally \emph{NOT} again for each of the first $n-1$ bits in order to return them to their original states. Figure \ref{t1} shows a circuit representation for $t_1$.

\begin{figure}[ht!]
\centering
\includegraphics[width=90mm]{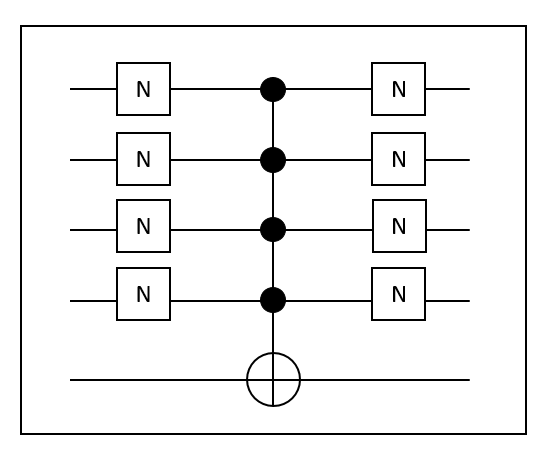}
\caption{Circuit representation for $t_1$. \label{t1}}
\end{figure}

The gate $t_2$ adds 1 to each element. In the original binary settings, the last bit is always inverted. The second to last bit is inverted if and only if there was a carry bit from the addition on the last bit. Similarly, the $i$th bit is inverted if and only if there was a carry bit from the $(i+1)$th bit.

More specifically, the $i$th bit is inverted if and only if all the lower bits are $1$, where $i = 1, 2, ... , n-1$. Thus $t_2$ can be implemented as a cascade of the following gates:

\begin{enumerate}
\item[(1)] \emph{NOT} on the last bit
\item[(2)] \emph{CNOT} on the last two bits, with the $(n-1)$th bit as the target line
\item[(3)] \emph{C$^2$NOT} on the last three bits, with the $(n-2)$th bit as the target line
\item[...]
\item[(n)] \cnnnot on all $n$ bits, with the highest bit as the target line 
\end{enumerate}

Figure \ref{t2} illustrates the construction on a 5-bit input.

\begin{figure}[ht!]
\centering
\includegraphics[width=90mm]{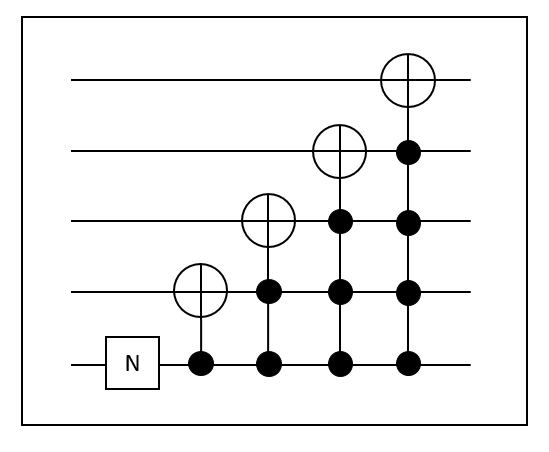}
\caption{Circuit representation for $t_2$. \label{t2}}
\end{figure}

\end{proof}

Combining the two lemmas above, we can conclude that any $n$-bit reversible gate can be constructed by cascading multi-control Toffoli gates with up to $n-1$ control lines.

\section{\cinot from Variated Toffoli}

In the previous section we constructed reversible gates using \cinot, so what remains to do is to construct \cinot from the 3-bit variated Toffoli gate. We first examine a few cases when $i$ is small.

\subsection{Base Cases}

When $i$ is $1$, \cinot becomes \emph{CNOT}, which can be implemented by applying the variated Toffoli twice. We apply the variated Toffoli in the same order as the input, with the third bit being the target line and the second bit being the inversion line. Suppose the input is $(a, b, c)$, then this will generate the following output:

$$(a, b, c) \rightarrow (a, b+1, ab+c) \rightarrow (a, b, a(b+1) + ab + c) = (a, b, a+c)$$

The resulting gate is equivalent to a \emph{CNOT} on the bits $a$ and $c$, with $c$ being the target line.

When $i$ is $0$, \cinot becomes \emph{NOT}, which can be implemented using a slightly more complicated procedure:

\begin{enumerate}
\item Apply variated Toffoli on the 1st, 2nd, and 3rd bits. Starting with input $(a, b, c)$, this gives us $(a, b+1, ab+c)$.
\item Apply variated Toffoli on the 2nd, 1st, and 3rd bits. This gives us $(a+1, b+1, a+c)$.
\item Apply variated Toffoli on the 1st, 2nd, and 3rd bits. This gives us $(a+1, b, a+c+(a+1)(b+1)) = (a+1, b, ab+b+c+1)$.
\item Apply variated Toffoli on the 2nd, 1st, and 3rd bits. This gives us $(a, b, ab+b+c+1+(a+1)b) = (a, b, c+1)$.
\end{enumerate}

The resulting gate is equivalent to a \emph{NOT} gate on $c$.

When $i$ is $2$, \cinot becomes the regular Toffoli gate, which can be implemented by a variated Toffoli followed by a \emph{NOT} on the second bit.

\subsection{Borrowed Bits}

Before constructing the general cases, we introduce a new concept called the \emph{borrowed bit}, which is a weaker version of the ancilla bit. The main difference between a borrowed bit and an ancilla bit is the specification on the input value.

\begin{Def}
A borrowed bit is an auxiliary input line provided to a gate such that it could be manipulated by the gate, but must be returned with the same value as its original state. A borrowed bit could start as either 0 or 1. The gate that requires a borrowed bit does not have the privilege of specifying its starting value.
\end{Def}

Borrowed bits are strictly less powerful than value-specified ancilla bits. However, their flexibility can make them useful when building reversible gates recursively. For example, gate $A$ requires a borrowed bit and changes it to an unknown intermediate stage during the computation. Gate $B$, which is a subroutine used as a building block for $A$, needs another borrowed bit. Then $B$ can simply reuse the bit from $A$, change it to anything, and revert it back to the original state before returning it back to $A$. The advantage comes from the fact that gate $B$ does not require the auxiliary bit to start with some particular value.

In practice, borrowed bits are also much cheaper than ancilla bits. Creating value-specific input lines such as ancilla bits might induce undesirable energy consumption. The borrowed bits, however, can simply be taken from unused memory. This significantly lowers the implementation cost.

\subsection{General Cases}

Now we examine how to construct \cinot where $i > 2$, using induction. In the previous section where we construct the reversible gates using \cinot, no ancilla bit is necessary. For the construction of \cinot, we aim at using only one ancilla bit. An ancilla bit is an auxiliary input line, so with the extra input bit, an $(i+1)$-bit \cinot can be expressed as an $(i+2)$-bit gate, where the first $(i+1)$ bits are used to perform \cinot, and the last bit remains unchanged.

The core idea of this construction is as follows: Assume we know a construction with at most one ancilla bit for gate \ciinot, then we have a method of constructing \cinot from the \ciinot gate using an additional ancilla bit. However, this alone does not guarantee that the total number of ancilla bits needed for \cinot is 1. The reason is that, the \ciinot gate is used within \cinot as a subroutine; the ancilla bit required by \ciinot requires a fixed input of 0 or 1, thus it cannot use the same ancilla bit as the one for the construction of \cinot, which is currently in an intermediate stage with an unknown value.

However, this issue can be resolved using the idea of borrowed bits. Assume that we have a construction for \ciinot using one borrowed bit, then we claim that we can construct \cinot using another borrowed bit. Since borrowed bits can be reused in subroutines, the total number of borrowed bits used to construct $\cinot$ would merely be 1. The construction of \cinot (an $(i+1)$-bit gate) is as follows, assuming the input is $(a_1, a_2, ... , a_i, a_{i+1}, x)$ with $a_{i+1}$ being the target line and $x$ being the borrowed bit:

\begin{enumerate}
\item[(1)] Apply \ciinot on the first $(i-1)$ bits and the last bit, with the last bit being the target line. This gives us $(a_1, a_2, ... , a_i, a_{i+1}, x + \prod_{k=1}^{i-1}{a_k})$. Denote the new value of the last bit ($x + \prod_{k=1}^{i-1}{a_k}$) by $x'$.
\item[(2)] Apply \emph{CCNOT} on bits $(a_i, x', a_{i+1})$, with $a_{i+1}$ being the target line. This gives us $(a_1, a_2, ... , a_i, a_{i+1} + a_i(x + \prod_{k=1}^{i-1}{a_k}), x')$, which simplifies to $(a_1, a_2, ... , a_i, a_{i+1} + a_ix + \prod_{k=1}^{i}{a_k}, x+\prod_{k=1}^{i-1}{a_k})$.
\item[(3)] Apply \ciinot on the first $(i-1)$ bits and the last bit again to return the borrowed bit to its original value. This gives us $(a_1, a_2, ... , a_i, a_{i+1} + a_ix + \prod_{k=1}^{i}{a_k}, x)$. Denote the new value of the target line ($a_{i+1} + a_ix + \prod_{k=1}^{i}{a_k}$) by $a'_{i+1}$.
\item[(4)] Apply \emph{CCNOT} on bits $(a_i, x, a'_{i+1})$ again, removing the redundant term in $a'_{i+1}$. This gives us $(a_1, a_2, ... , a_i, a_{i+1} + \prod_{k=1}^{i}{a_k}, x)$.
\end{enumerate}

Figure \ref{cinot} illustrates this construction on a 5-bit input, with the last line being the borrowed bit.

\begin{figure}[ht!]
\centering
\includegraphics[width=90mm]{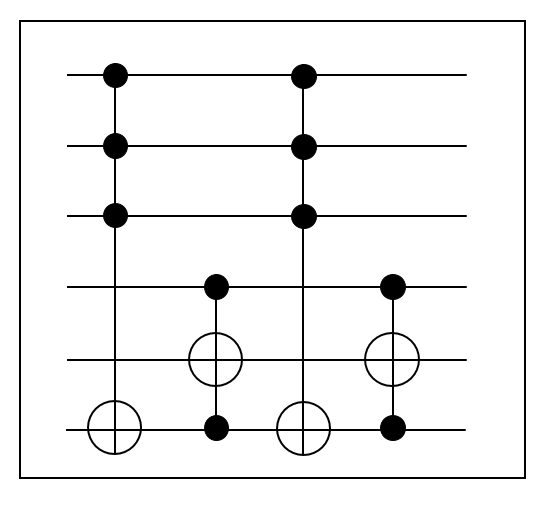}
\caption{Construction of \cfnot from \emph{C\textsuperscript{3}NOT}. \label{cinot}}
\medskip
\small
The last line is the borrowed bit, and the second to last line is the target line.
\end{figure}

The resulting gate is equivalent to applying \cinot on the first $(i+1)$ bits. The borrowed bit $x$ participated in the computation, but was returned to its original value eventually. A small subtlety is that we used \emph{CCNOT}, which is the regular Toffoli gate. Note that Toffoli can be constructed from the variated Toffoli without any borrowed bits as explained in Section 2.4.1.

By induction, we have constructed \cnnot using only 1 borrowed bit.

\section{Conclusion}

Combining the two previous sections, we provided a construction of any $n$-bit reversible gate using the variated Toffoli gate plus at most 1 borrowed bit. We also saw from previous chapters that a small gate like the variated Toffoli cannot possibly generate all $n$-bit reversible gates using pure cascades with no additional input lines. Thus, in terms of minimizing the number of ancilla bits needed for $n$-reversible gate construction, we have found an optimal solution. Compared to the other optimal construction method proposed by Shende et al. \cite{shende02}, our construction uses only a single base gate.

\chapter{Constructing Even Permutations without Ancilla Bits}

In the previous chapter we discussed the variated Toffoli gate, which can be used to generate any $n$-bit reversible gate using one borrowed bit. One extra input line is optimal for general universality, but for a special subset of reversible gates, we can achieve even better results. In this chapter, we explore the set of reversible gates that are equivalent to even permutations, and prove that the variated Toffoli can generate any such gate using no extra input line.

Lemma 5 states that any permutation can be broken into one of the two basic types of transformations $t_1$ and $t_2$. To make analysis easier for this chapter, we slightly tweak the definition of $t_1$ without changing the correctness.

\begin{Lem}
Any permutation $\pi$ on the set $N = \{0, 1, 2, ... , 2^n-1\}$ can be decomposed into the following two basic transformations:
\begin{enumerate}
\item[(1)] $t'_1$: swapping $2^n-2$ and $2^n-1$ 
\item[(2)] $t'_2$: shifting $k$ to $(k+1) \Mod 2^n$, $\forall k \in N$
\end{enumerate}
Furthermore, if $\pi$ is an even permutation, then the number of $t'_1$ and the number of $t'_2$ in the decomposition are both even numbers.
\end{Lem} 

See Appendix B for the proof of this lemma.

Given any even reversible gate $g$, consider the sequence of $t'_1$ and $t'_2$ that its corresponding permutation is broken down into. Since there are an even number of $t'_1$ and $t'_2$ respectively, we can break down the sequence into pairs of adjacent transformations. Thus we only need to figure out how to construct gates corresponding to each possible pair of transformations using the variated Toffoli gate.

We build the gates recursively. In the previous chapter we already saw how the variated Toffoli gate can generate \emph{NOT}, \emph{CNOT}, and \emph{CCNOT} without any extra input lines. We also proved that the set $T_n$ can generate all $n$-bit reversible gates, so the variated Toffoli can generate any 3-bit gate without any extra input lines. We use this as our base case.

Assume that a construction exists for any $(i-1)$-bit gate that is an even permutation. Now we aim at constructing the $i$-bit even permutation gates. We consider the sequence of \T and \TT broken into adjacent pairs; these pairs fall into one of the four possible cases: (1) \T\T (2) \TT\TT (3) \T\TT (4) \TT\T. For convenience, we consider these combined transformations as matrices and denote them as $M_1$, $M_2$, $M_3$, and $M_4$ respectively.

In the first case, the two \T cancel out, so $M_1$ is the identity matrix.

In the second case, the combined transformation corresponds to adding 2 to any input, which is equivalent to leaving the last bit unchanged and adding 1 to the first $(i-1)$ bits. This is an even permutation gate with $i-1$ bits, so it is already handled based on our assumption.

The third case and the fourth case are similar. As an example, we show the construction for the third case. $M_3$ consists of \T followed by \TT. In the previous chapter, we saw that \TT can be broken into a cascade of \emph{NOT}, \emph{CNOT}, \emph{CCNOT}, ... , \ciinot. The highest bit can be used as a borrowed bit for all but the last \ciinot in the sequence, so it remains to be solved how to use the variated Toffoli gate to generate \T followed by \ciinot with no extra input line. This \ciinot operates on all the bits, with the highest bit being the target line. Notice that \T, the swap of $2^i-2$ and $2^i-1$, is equivalent to a \ciinot on all $i$ bits, with the lowest bit being the target line.

Suppose that the input is $(a, c_1, c_2, ... , c_{i-2}, b)$, then the gate corresponding to these two \ciinot cascaded together will change the highest bit $a$ to $a + a\prod_{k=1}^{i-2}{c_k} + b\prod_{k=1}^{i-2}{c_k}$, and change the lowest bit $b$ to $b+a\prod_{k=1}^{i-2}{c_k}$.

In order to construct such a gate using smaller gates, we consider breaking the elements $(c_1, c_2, ... , c_{i-2})$ into roughly even halves. Denote the first half by $x$ and the second half by $y$. Let $X$ be the product of all elements in $x$ and $Y$ be the product of all elements in $y$. We can construct the gate using a sequence of \cknot:

\begin{eqnarray*}
&&(a, x, y, b) \\
&\rightarrow& (a+bX, x, y, b) \\
&\rightarrow& (a+bX, x, y, b+aY+bXY) \\
&\rightarrow& (a+aXY+bXY, x, y, b+aY+bXY) \\
&\rightarrow& (a+aXY+bXY, x, y, b+aXY) \\
\end{eqnarray*}

Notice that each of these \cknot gates has fewer than $i$ control lines, which are already handled based on the inductive assumption. This completes the construction for any conservative reversible gate given all even permutation gates with up to $i-1$ bits. Thus, the variated Toffoli gate can be used to generate any even permutation reversible gate using no extra input lines.

\chapter{Conservative Universality Using One Ancilla Bit}

Conservative logic gates are a subset of logic gates with practical applications in modeling physical systems with quantity conservation. The Fredkin gate is a common conservative reversible gate. In this chapter we examine a way to construct all conservative reversible gates using 3-bit Fredkin gates plus one ancilla bit.

Fredkin gates are conservative, and thus can only generate conservative transformations.

We can divide the set $I_m = \{0, 1\}^m$ into $m+1$ disjoint subsets $S_0, S_1, ... , S_m$, where $S_i$ represents the subset of all such strings with Hamming weight $i$. $S_i$ contains ${m}\choose{i}$ elements.

If we consider a conservative reversible gate as a permutation $\pi$ on $I_m$, then due to the fact that conservative gates preserve Hamming weights, $\pi$ must not map an element in subset $S_i$ to another element in a different subset $S_j$. Thus, we can decompose $\pi$ into permutations within each of the subsets. That is, $$\pi = \pi_0 \cdot \pi_1 \cdot ... \cdot \pi_m$$ where $\pi_i$ denotes the permutation on $S_i$. This representation for conservative reversible gates will be used throughout this chapter.

\section{From \ciswap to Conservative Reversible Gates}

We consider breaking down the gate construction into permutations, using \ciswap to construct each $\pi_k$, $k = 0, 1, ... , m$. Basically we aim at constructing any transposition within $S_k$ using \ciswap. $S_k$ consists of strings with Hamming weight $k$. Thus, we can use \ckkswap as a transposition of any two strings in $S_k$ with a Hamming distance of 2. However, \ckkswap also takes effect on strings in $S_{k+2}, S_{k+2}, ... , S_{m-1}$. In order to not mess up the permutations on the subsets that have already been constructed, we build the permutations starting from subsets with lower Hamming weight.

Now what remains is to construct $\pi_k$ using \ckkswap. It suffices to use \ckkswap to construct any transposition that is part of $\pi_k$. Take $S_{m-2}$ as an example, and consider the transposition on two strings $s_{a, b}$ and $s_{c, d}$, where $s_{i, j}$ indicates a string with locations $i$ and $j$ being 0. Consider an intermediate string $s_{b, c}$. Then we have the following procedure to swap $s_{a, b}$ and $s_{c, d}$, propagating through the intermediate string:

\begin{enumerate}
\item[(1)] Apply \cmmswap, using bits other than $a, b, c$ as control lines, and bits $a, c$ as the target lines. This swaps $s_{a, b}$ and $s_{b, c}$.
\item[(2)] Apply \cmmswap, using bits other than $b, c, d$ as control lines, and bits $b, d$ as the target lines. This swaps $s_{b, c}$ and $s_{c, d}$.
\end{enumerate}

This is based on the assumption that $b \neq c$. When $b = c$, the conversion can be done using a single \cmmswap.

After these steps, $s_{a, b}$ becomes $s_{c, d}$, while $s_{c, d}$ is now $s_{b, c}$. We need one more step to propagate $s_{c, d}$ back to $s_{a, b}$:

\begin{enumerate}
\item[(3)] Apply \cmmswap, using bits other than $a, b, c$ as control lines, and bits $a, c$ as the target lines. This swaps $s_{a, b}$ and $s_{b, c}$.
\end{enumerate}

After these steps, we will have swapped $s_{a, b}$ and $s_{c, d}$ without changing any other input in $S_{m-2}$.

In the general case of $S_k$ where we want to construct the transposition of two strings $s_1, s_2$, we can always find a sequence of intermediate strings between $s_1$ and $s_2$, such that each pair of adjacent strings have Hamming distance 2. Thus we can use \ckkswap to swap any pair of adjacent strings.

The first phase involves shifting $s_1$ all the way to $s_2$. This converts $s_1$ to $s_2$ while shifting all other strings in the sequence to the left by 1. The second phase involved shifting $s_2$ all the way back to $s_1$. This converts $s_2$ to $s_1$ while shifting all other strings back to their original locations. Figure \ref{swap} illustrates the two phases when $s_1$ and $s_2$ are Hamming distance $10$ away.

\begin{figure}[ht!]
\centering
\includegraphics[width=90mm]{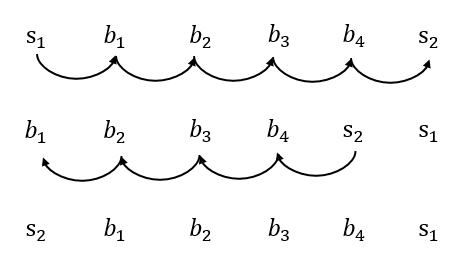}
\caption{Two phases constructing the transposition between two strings. \label{swap}}
\medskip
\small
The last line is the borrowed bit, and the second to last line is the target line.
\end{figure}

Thus we have a construction of any conservative reversible gate using \ckswap.

\section{Constructing \ciswap}

We consider a recursive implementation of \ciswap, similar to the process for constructing \cinot. We explore two methods of constructing \ciswap from \emph{C\textsuperscript{i-1}SWAP}. The first method uses an ancilla bit, while the second method uses a borrowed pair of bits that are opposite. The final construction for \ciswap merges the two methods and uses one ancilla bit in total.

\subsection{The Ancilla Bit Construction}

This construction requires an ancilla bit with its original value set to 0. Assume the input is $(a_1, a_2, ... , a_{i+1}, a_{i+2}, x)$ with $(a_{i+1}, a_{i+2})$ being the target lines and $x$ being the ancilla bit, whose initial value is set to 0. We also assume that there exists a construction for \emph{C\textsuperscript{i-1}SWAP}, disregarding the extra input requirements for this subroutine. The construction is as follows:

\begin{enumerate}
\item[(1)] Apply \ciiswap on the first $i$ bits and the last bit $x$, with $a_{i}$ and $x$ being the target lines. After this step, the last bit becomes $a_i$ if all but the last control lines are set to 1. It remains $0$ otherwise.
\item[(2)] Apply \cswap on the last three bits, with $a_{i+1}$ and $a_{i+2}$ being the target lines, and the last bit being the control line.

When there is at least one 0 in the set $(a_1, a_2, ... , a_{i-1})$, the control line will be 0, and thus the target lines will not be swapped. Otherwise if $(a_1, a_2, ... , a_{i-1})$ are all 1, the control line will now be $a_i$, and thus the target lines are changed if and only if $a_i$ is also 1. This effectively swaps the target lines if and only if all control bits are $1$.
\item[(3)] Apply \ciiswap on the first $i$ bits and the last bit again. The ancilla bit is reverted to its original value because the first $i$ bits were never changed.
\end{enumerate}

Notice that this alone does not give us a single-ancilla-bit construction for \emph{C\textsuperscript{i}SWAP}, because the ancilla bit needed during this step cannot be reused in the subroutines.

\subsection{The Borrowed Pair Construction}

This construction requires two extra input lines whose original values are set to 0 and 1. Assume the input is $(a_1, a_2, ... , a_{i+1}, a_{i+2}, x, y)$ with $(a_{i+1}, a_{i+2})$ being the target lines and $(x, y)$ being the borrowed pair, where $x$ is set to 0 and $y$ is set to 1. The construction is as follows:

\begin{enumerate}
\item[(1)] Apply \emph{CSWAP} on bits $(a_1, x, y)$, with $x$ being the control line and $(a_1, y)$ being the target lines. This changes $y$ to $y' = y(1-x) + a_1x$, and changes $a_1$ to $a'_1 = a_1(1-x)+yx$, which simplifies to $a_1 + a_1x$ because $x$ and $y$ take opposite values.
\item[(2)] Apply \ciiswap on the first $i$ bits and $x$, with $a_i$ and $x$ being the target lines. The new value for $a_i$ after this step is:
\begin{eqnarray*}
a'_i &=& a_i (1-\prod_{k=1}^{i-1}{a_k}) + x \prod_{k=1}^{i-1}{a_k} \\
&=& a_i + (a_1 + a_1x) \prod_{k=2}^{i}{a_k} + x (a_1 + a_1x) \prod_{k=2}^{i-1}{a_k} \\
&=& a_i + \prod_{k=1}^{i}{a_k} + x \prod_{k=1}^{i}{a_k}
\end{eqnarray*}
\item[(3)] Apply \emph{CSWAP} on bits $(a'_i, a_{i+1}, a_{i+2})$, with $(a_{i+1}, a_{i+2})$ as the target lines. At this point, the swap of $(a_{i+1}, a_{i+2})$ depends on the value $a_i + \prod_{k=1}^{i}{a_k} + x \prod_{k=1}^{i}{a_k}$.
\item[(4)] Apply \ciiswap on the first $i$ bits and $x$ once again to revert $a'_i$ and $x'$ to their original values $a_i$ and $x$.
\item[(5)] Apply \emph{CSWAP} on bits $(a_1, x, y)$ again to revert $y'$ and $a'_1$ to their original values $y$ and $a_1 $.
\item[(6)] - (10) Repeat gates (1) to (5), replacing $x$ with $y$, and $y$ with $x$. This initiates another swap on the bits $(a_{i+1}, a_{i+2})$, controlled by the value $a_i + \prod_{k=1}^{i}{a_k} + y \prod_{k=1}^{i}{a_k}$
\end{enumerate}

After cascading the gates above, the target lines $(a_{i+1}, a_{i+2})$ are conditionally swapped twice, using two different control values $a_i + \prod_{k=1}^{i}{a_k} + y \prod_{k=1}^{i}{a_k}$ and $a_i + \prod_{k=1}^{i}{a_k} + x \prod_{k=1}^{i}{a_k}$. This is equivalent to a single conditional swap on the target lines with the control bit being the sum of these two values, which is $(x+y)\prod_{k=1}^{i}{a_k}$. Because $x$ and $y$ take opposite values, $x+y$ must be $1$, so swapping $(a_{i+1}, a_{i+2})$ is effectively controlled by $\prod_{k=1}^{i}{a_k}$, which includes all the control lines.

Notice that in the calculation of the final control bit value, the only requirement on $(x, y)$ is actually the fact that they are opposite. It is not hard to verify that if $x$ is set to 1 and $y$ is set to 0, the same argument still holds. This is why we call it a borrowed pair; the requirement is still weaker than a pair of ancilla bits.

\subsection{Merged construction for \ciswap}

The final construction of \ciswap makes use of both constructions. The core idea of this combination step is based on the following observation: In the borrowed pair construction, when the borrowed pair of bits $(x, y)$ share the same initial value, then the output will be exactly the same as the input, failing our original intention of swapping $a_{i+1}$ and $a_{i+2}$ controlled by $\prod_{k=1}^i{a_k}$. However, if we use the target lines as the borrowed bits, then when the target lines share the same initial value, whether the gate swaps them or not has no effect on the final output.

This observation provides us a way to incorporate both constructions in one. When constructing \ciswap from \emph{C\textsuperscript{i-1}SWAP}, we inevitably need to pick one of the two constructions for \emph{C\textsuperscript{i-1}SWAP}. In order to reuse bits to the largest extent, we choose the borrowed pair construction, using the two target lines $(a_{i+1}, a_{i+2})$ in the \ciswap as the borrowed pair for \emph{C\textsuperscript{i-1}SWAP}. We use this method to construct \ccswap from \emph{CSWAP}, \cccswap from \emph{CCSWAP}, and so on. In this way, the only extra input lines we need from outside are from the last layer when we construct \ciswap from \emph{C\textsuperscript{i-1}SWAP}, which are a borrowed pair $(x, y)$ whose values are opposite. This gives us a construction of \ciswap using Fredkin gates with 2 ancilla bits. These are ancilla bits instead of borrowed bits because their initial values cannot be completely arbitrary.

However, with the ancilla bits construction, we can further reduce the number of ancilla bits needed. We use the borrowed pair construction for constructing up to \emph{C\textsuperscript{i-1}SWAP}, using the same target line trick described in the previous paragraph. Then, when constructing \ciswap from \emph{C\textsuperscript{i-1}SWAP}, we switch to the ancilla bit construction, using an ancilla bit with its value set to 0. This gives us a construction of \ciswap using Fredkin gates with only 1 ancilla bit.

Combined with the fact that any conservative reversible gate can be generated using \ckswap without any extra input lines, we now have a way to construct any $n$-bit conservative reversible gate using Fredkin gates and 1 ancilla bit. In the next section, we examine if it is possible to do even better.

\section{Necessity of Ancilla Bits} We mentioned in previous chapters that a borrowed bit is weaker than an ancilla bit in the sense that it does not guarantee a fixed input. In the previous section, we provided a construction of any $n$-bit conservative reversible gate using the Fredkin gate plus 1 ancilla bit. But since borrowed bits are cheaper, is it possible to find a construction using only borrowed bits?

\begin{Cla}
Even with an unlimited number of borrowed bits, we still cannot generate all $n$-bit conservative reversible gates with 3-bit Fredkin gates.
\end{Cla}

\begin{proof}
Assume we have $b$ borrowed bits. The statement is thus equivalent to the statement that Fredkin gates cannot be used to generate all $n$-bit conservative reversible gates in the $(n+b)$-bit universe. For convenience, let $m$ denote $n+b$.

Using the permutation decomposition technique we introduced at the beginning at this chapter, we study the parity of the permutation $\pi_i$. Take \emph{SWAP} as an example. $S_0$ and $S_m$ contain only one element each, so \emph{SWAP} is an even permutation on these two subsets. On $S_1$, \emph{SWAP} is an odd permutation because it swaps two elements and keeps everything else unchanged. On $S_i$, without loss of generality we assume that \emph{SWAP} operates on the first $i$ bits. Thus it pairs up two strings that start with $01$ and $10$ but all of the remaining $m-2$ bits are the same, and performs a transposition on those two numbers. $\pi_2$ is equivalent to the product of all such transpositions. Since the Hamming weight is $i$, there must be exactly $i-1$ 1's in the remaining $n-2$ bits. Thus there are a total of ${m-2}\choose {i-1}$ pairs, corresponding to ${m-2} \choose {i-1}$ transpositions. Whether \emph{SWAP} is an even or permutation on $S_i$ depends on the parity of this number.

We define a \emph{parity vector} for each conservative reversible gate. A parity vector $p = (p_0, p_1, ..., p_m)$ represents the permutation parity in each of the $m$ subsets. Specifically, $p_i$ is 1 if $\pi_i$ is an odd permutation, 0 if even. The parity vector for \emph{SWAP} would then be $$(0, 1, {{m-2} \choose 1} \Mod 2, {{m-2} \choose 2} \Mod 2, ... , {m-2 \choose m-3} \Mod 2, 1, 0)$$ Similarly, the parity vector for \emph{CSWAP} would be $$(0, 0, 1, {m-3 \choose 1} \Mod 2, {{m-3} \choose 2} \Mod 2, ..., {m-3 \choose m-4} \Mod 2, 1, 0)$$ The parity vector for \ckswap is $$(0, ..., 0, 1, {m-2-k \choose 1} \Mod 2, {m-2-k \choose 2} \Mod 2, ..., {m-2-k \choose m-3-k} \Mod 2, 1, 0)$$ These parity vectors put together are equivalent to the $\Mod 2$ remainder of the elements in the Pascal triangle.

Now we can discuss whether a subset $G$ of the gates ${\emph{SWAP}, \cswap, ... , \ckswap}$ can generate a gate that is not in set $G$. It is not hard to see that if gates $g_1$ and $g_2$ have parity vectors $p_1$ and $p_2$ correspondingly, then the gate generated by cascading $g_2$ after $g_1$ would have a parity vector of $p_1 + p_2$. In fact any gate generated by cascading any number of $g_1$ and $g_2$ in any order would have a parity vector that is a linear combination of $p_1$ and $p_2$.

We denote by $p_{si}$ the parity vector of \ciswap. From our analysis above we see that $p_{s2}$ starts with three 0's and a 1; it ends with a 1 and 0. Assume that there is a way to generate \emph{CCSWAP} by cascading \emph{SWAP} and \emph{CSWAP}, then there exists two integers $a$ and $b$ such that $$p_{s2} = a\cdot p_{s0} + b\cdot p_{s1}$$ The first few elements of the right-hand side are $0, a \Mod 2, (am + b) \Mod 2, ...$, while the first few elements of the left-hand side are $0, 0, 0, 1, ...$. Thus we have $$a \Mod 2 = 0$$ and $$(am+b) \Mod 2 = 0$$ These two equations give us $a \Mod 2 = 0$ and $b \Mod 2 = 0$. That is to say, if there exists a way to generate \emph{CCSWAP} with \emph{SWAP} and \cswap, then there must be an even number of each. However, this also means that $a\cdot p_{s0} + b \cdot p_{s1}$ is a 0 vector. This is a contradiction since the fourth element of $p_{s2}$ has to be a 1.

We have just proven that cascades of \emph{SWAP} and \cswap cannot generate \emph{CCSWAP}, no matter how large $m$ is. In fact, we can easily generalize the result to the following: Given all the \ciswap with $i$ up to $k-1$, we cannot generate \ckswap through cascading, no matter how many borrowed bits are provided.

\end{proof}

\chapter{Conclusions and Further Work}

There are many ways to evaluate how good a certain gate construction is. The main problem we are concerned with in this thesis is the number of ancilla bits used. Other metrics include the gate complexity, circuit depth, and others.

We introduced the concept of borrowed bits, which allows us to consider a weaker but also cheaper version of ancilla bits. In terms of the number of extra input lines used, we arrived at three optimal constructions. We invented a single gate that can be used to generate both any $n$-bit reversible gate using only one borrowed bit and any $n$-bit even permutation reversible gate using no extra input line. We also designed a construction for any $n$-bit conservative reversible gate using the Fredkin gate with only 1 ancilla bit. These results were proved to have reached the minimum number of necessary extra input lines.

In terms of other metrics, these constructions are not optimal. For example, the gate complexity (asymptotic number of basic gates used) for generating \cnnot using Toffoli and $O(\log n)$ ancilla bits can be polynomial. However, in our construction using the variated Toffoli and 1 borrowed bit, we need an exponential gate complexity due to the fact that we construct \cinot from \ciinot, rather than construct it using a binary recursion.

An interesting question to ask is whether there exists a construction using the same number of extra input lines, while still achieving polynomial gate complexity for \emph{C\textsuperscript{n}NOT}; or is it possible that there exists a trade-off relationship between the number of ancilla bits used and the best gate complexity that can be achieved? If so, then what is the minimum number of ancilla bits required while achieving a low enough gate complexity, for example, asymptotically the same as the lower bound predicted by Shannon's counting argument? These questions remain to be answered.
\appendix
\chapter{Small Gate Cannot Be Universal Without Ancilla Bits}

For completeness, we prove the following statement which has already been proved by Toffoli \cite{toffoli80}: No reversible gate with a constant number of input bits can be universal without extra input lines.

\begin{proof}
Assume that there exists a $k$-bit reversible gate $g$ that is universal without extra input lines. We decompose $g$ into transpositions $\tau_1 \tau_2 ... \tau_j$. 

Consider an $n$-bit gate where $n>k$. Then for an input with $n$ bits, $g$ is equivalent to an $n$-bit gate that operates on $k$ bits and leaves the other $n-k$ bit unchanged. Thus, each $\tau_i$ corresponds to a transposition in this $n$-bit version of $g$ that has the same effect on the $k$ effective bits, but leave the other $n-k$ bits unchanged. Since the unchanged $n-k$ bits can have $2^{n-k}$ possibilities, the total number of transpositions in this $n$-bit version of $g$ must be a multiple of $2^{n-k}$, which is an even number since $n>k$. Thus $g$ as an $n$-bit gate is an even permutation. 

Any cascade of even permutations is still an even permutation, meaning that no odd $n$-bit gate can be generated by cascading $g$. Thus $g$ must not be universal.
 
\end{proof}

\clearpage
\newpage

\chapter{Decomposing Permutation into Basic Operations}

We prove the following lemma: Any permutation $\pi$ on the set $N = \{0, 1, 2, ... , 2^n-1\}$ can be decomposed into the following two basic operations:
\begin{enumerate}
\item[(1)] $t_1$: swapping 0 and 1 
\item[(2)] $t_2$: shifting $k$ to $(k+1) \Mod 2^n$, $\forall k \in N$
\end{enumerate}

\begin{proof}
Any permutation can be decomposed into a sequence of transpositions. For transposition $(i, j)$, we can further decompose it into adjacent swaps: 

\begin{itemize}
\item Swap $(i, i+1)$.
\item Swap $(i, i+2)$.
\item ...
\item Swap $(i, j)$.
\item Swap $(j, j-1)$.
\item Swap $(j, j-2)$.
\item ...
\item Swap $(j, i+1)$.
\end{itemize}

Thus it suffices to find a construction for any adjacent swap $(i, i+1)$ using $t_1$ and $t_2$.

\begin{itemize}
\item Apply $t_2$ for $2^n-i$ times. This shifts everything to the right for $2^n-i$ times, placing $(i, i+1)$ at $(0, 1)$.
\item Apply $t_1$. This swaps $(0, 1)$, which corresponds to $(i, i+1)$ originally.
\item Apply $t_2$ for $i$ times. This shifts everything back to their original position, except that $(i, i+1)$ are now swapped.
\end{itemize}

\end{proof}

We can also count the total number of $t_1$ and $t_2$ used in this procedure. For each adjacent swap, $t_1$ is applied once, and $t_2$ is applied $2^n$ times.

Each transposition is always decomposed into an odd number of adjacent swaps. Thus for each transposition, $t_1$ is applied an odd number of times, and $t_2$ is applied an even number of times.

If the permutation is even, then it can be decomposed into an even number of transpositions. Thus the total number of $t_1$ for this even permutation is even, and so is the total number of $t_2$. Changing $t_1$ to $t'_1$ does not change this fact.
\begin{singlespace}
\bibliography{main}
\bibliographystyle{abbrv}

\end{singlespace}

\end{document}